\documentclass[11pt,leqno]{article}

\input{style}
\usepackage{enumitem}

\numberwithin{equation}{section}

\begin{document}

\title{Muckenhoupt's $\Ap{p}$ condition and the existence of the
  optimal martingale measure}

\author{Dmitry Kramkov\thanks{The author also holds a part-time
    position at the University of Oxford. This research was supported
    in part by the Oxford-Man Institute for Quantitative Finance at
    the University of Oxford.}  \;and Kim Weston,\\
  Carnegie Mellon University, \\
  Department of Mathematical Sciences,\\
  5000 Forbes Avenue, Pittsburgh, PA, 15213-3890, US}

\date{\today}

\maketitle
\begin{abstract}
  In the problem of optimal investment with utility function defined
  on $(0,\infty)$, we formulate sufficient conditions for the dual
  optimizer to be a uniformly integrable martingale.  Our key
  requirement consists of the existence of a martingale measure whose
  density process satisfies the probabilistic Muckenhoupt $\Ap{p}$
  condition for the power $p=1/(1-a)$, where $a\in (0,1)$ is a lower
  bound on the relative risk-aversion of the utility function. We
  construct a counterexample showing that this $\Ap{p}$ condition is
  sharp.
\end{abstract}

\begin{description}
\item[Keywords:] utility maximization, optimal martingale measure, BMO
  martingales, $\Ap{p}$ condition.
\item[AMS Subject Classification (2010):] 60G44, 91G10.
\end{description}

\section{Introduction}
\label{sec:introduction}

An unpleasant qualitative feature of the general theory of optimal
investment with a utility function defined on $(0,\infty)$ is that the
dual optimizer $\widehat Y$ may not be a uniformly integrable
martingale. In the presence of jumps, it may even fail to be a local
martingale. The corresponding counterexamples can be found in
\cite{KramSch:99}. In this paper, we seek to provide conditions under
which the uniform martingale property for $\widehat Y$ holds and thus,
$\widehat Y/\widehat Y_0$ defines the density process of the
\emph{optimal} martingale measure $\widehat{\mathbb{Q}}$.

The question of whether $\widehat{Y}$ is a uniformly integrable
martingale is of longstanding interest in mathematical finance and can
be traced back to \cite{HePears:91b} and
\cite{KaratLehocShrXu:91}. This problem naturally arises in situations
involving utility-based arguments. For instance, it is relevant for
pricing in incomplete markets, where according to
\cite{HugonKramSch:05} the existence of $\widehat{\mathbb{Q}}$ is
equivalent to the uniqueness of marginal utility-based prices for
every bounded contingent claim.

Our key requirement consists of the existence of a dual
supermartingale $Z$, which satisfies the probabilistic Muckenhoupt
$\Ap{p}$ condition for the power $p>1$ such that
\begin{equation}
  \label{eq:1}
  p=\frac1{1-a}.
\end{equation}
Here $a\in (0,1)$ is a lower bound on the relative risk-aversion of
the utility function.  As we prove in Theorem~\ref{th:2}, this
condition, along with the existence of an upper bound for the relative
risk-aversion, yields $\Ap{p'}$ for $\widehat Y$ for some $p'>1$. This
property in turn implies that the dual minimizer $\widehat Y$ is of
class $(\mathbf{D})$, that is, the family of its values evaluated at
all stopping times is uniformly integrable.  In
Proposition~\ref{prop:3}, we construct a counterexample showing that
the bound~\eqref{eq:1} is the best possible for $\widehat Y$ to be of
class $(\mathbf{D})$ even in the case of power utilities and
continuous stock prices.

A similar idea of passing regularity from some dual element to the
optimal one has been employed in \cite{DMSSS:97}, \cite{GranKraw:99}
and \cite{DGRSSS:02} for respectively, quadratic, power and
exponential utility functions defined on the whole real line. These
papers use appropriate versions of the Reverse H\"{o}lder $\Rq{q}$
inequality which is dual to $\Ap{p}$.  Note that contrary to $\Ap{p}$,
the uniform integrability property is not implied but rather
\emph{required} by $\Rq{q}$. While this requirement is not a problem
for real-line utilities, where the optimal martingale measures always
exist, it is clearly an issue for utility functions defined on
$(0,\infty)$.

Even if the dual minimizer $\widehat Y$ is of class $(\mathbf{D})$, it
may not be a martingale, due to the lack of the local martingale
property; see the single-period example for logarithmic utility in
\cite[Example~5.1$'$]{KramSch:99}. In Proposition~\ref{prop:2} we
prove that every \emph{maximal} dual supermartingale (in particular,
$\widehat Y$) is a local martingale if the ratio of any two positive
wealth processes is $\sigma$-bounded.

Our main results, Theorems~\ref{th:2} and \ref{th:3}, are stated in
Section~\ref{sec:exist-optimal-mart}.  They are accompanied by
Corollaries~\ref{cor:1} and~\ref{cor:2}, which exploit well known
connections between the $\Ap{p}$ condition and $\bmo$ martingales.

\section{Setup}
\label{sec:setup}

We use the same framework as in \cite{KramSch:99, KramSch:03} and
refer to these papers for more details. There is a financial market
with a bank account paying zero interest and $d$ stocks. The process
of stocks' prices $S = (S^i)$ is a semimartingale with values in
$\mathbf{R}^d$ on a filtered probability space $(\Omega, \mathcal{F},
(\mathcal{F}_t)_{t\in [0,T]}, \mathbb{P})$.  Here $T$ is a finite
maturity and $\mathcal{F} = \mathcal{F}_T$, but we remark that our
results also hold for the case of infinite maturity.

A (self-financing) portfolio is defined by an initial capital $x\in
\mathbf{R}$ and a predictable $S$-integrable process $H=(H^i)$ with
values in $\mathbf{R}^d$ of the number of stocks. Its corresponding
wealth process $X$ evolves as
\begin{displaymath}
  X_t = x + \int_{0} ^{t} H_u dS_u,\quad t\in [0,T].
\end{displaymath}

We denote by $\mathcal{X}$ the family of non-negative wealth
processes:
\begin{displaymath}
  \mathcal{X} \set  \descr{X\geq 0}{X \text{ is a wealth process}} 
\end{displaymath}
and by $\mathcal{Q}$ the family of equivalent local martingale
measures for $\mathcal{X}$:
\begin{displaymath}
  \mathcal{Q} \set \descr{\mathbb{Q}\sim \mathbb{P}}{\mtext{every}
    X\in \mathcal{X} \mtext{is a local
      martingale under} \mathbb{Q}}. 
\end{displaymath}
We assume that
\begin{equation}
  \label{eq:2}
  \mathcal{Q} \neq \emptyset, 
\end{equation}
which is equivalent to the absence of arbitrage;
see~\cite{DelbSch:94,DelbSch:98}.

There is an economic agent whose preferences over terminal wealth are
modeled by a utility function $U$ defined on $(0,\infty)$. We assume
that $U$ is of \emph{Inada type}, that is, it is strictly concave,
strictly increasing, continuously differentiable on $(0,\infty)$, and
\begin{displaymath}
  U'(0) = \lim_{x\to 0} U'(x) = \infty, \quad 
  U'(\infty) = \lim_{x\to \infty} U'(x) = 0. 
\end{displaymath}
For a given initial capital $x>0$, the goal of the agent is to
maximize the expected utility of terminal wealth.  The value function
of this problem is denoted by
\begin{equation}
  \label{eq:3}
  u(x) = \sup_{X\in \mathcal{X}, \; X_0 = x} \EP{U(X_T)}.
\end{equation}

Following \cite{KramSch:99}, we define the dual optimization problem
to~\eqref{eq:3} as
\begin{equation}
  \label{eq:4}
  v(y) = \inf_{Y\in{\cal Y}, Y_0=y}\EP{V(Y_T)},\quad y>0,  
\end{equation}
where $V$ is the convex conjugate to $U$:
\begin{displaymath}
  V(y)=\sup_{x>0} \left\{U(x)-xy\right\}, \quad y>0,
\end{displaymath}
and $\mathcal{Y}$ is the family of ``dual'' supermartingales to
$\mathcal{X}$:
\begin{displaymath}
  \mathcal{Y}=\descr{Y\geq 0}{XY \mtext{is a supermartingale for
      every} X\in\mathcal{X}}. 
\end{displaymath}
Note that the set $\mathcal{Y}$ contains the density processes of all
$\mathbb{Q}\in\mathcal{Q}$ and that, as $1\in \mathcal{X}$, every
element of $\mathcal{Y}$ is a supermartingale.

It is known, see \cite[Theorem~2]{KramSch:03}, that under~\eqref{eq:2}
and
\begin{equation}
  \label{eq:5}
  v(y) < \infty, \quad y>0, 
\end{equation}
the value functions $u$ and $-v$ are of Inada type, $v$ is the convex
conjugate to $u$, and
\begin{equation}
  \label{eq:6}
  v(y) = \inf_{\mathbb{Q}\in
    \mathcal{Q}}\EP{V\left(y\frac{d\mathbb{Q}}{d\mathbb{P}}\right)},\quad
  y>0. 
\end{equation}
The solutions $X(x)$ to \eqref{eq:3} and $Y(y)$ to \eqref{eq:4}
exist. If $y=u'(x)$ or, equivalently, $x=-v'(y)$, then
\begin{displaymath}
  U'(X_T(x)) = Y_T(y), 
\end{displaymath}
and the product $X(x)Y(y)$ is a uniformly integrable martingale.

The last two properties actually characterize optimal $X(x)$ and
$Y(y)$. For convenience of future references, we recall this
``verification'' result.

\begin{Lemma}
  \label{lem:1}
  Let $\widehat X\in \mathcal{X}$ and $\widehat Y\in \mathcal{Y}$ be
  such that
  \begin{displaymath}
    U'(\widehat{X}_T) = \widehat{Y}_T, \quad
    \EP{V(\widehat{Y}_T)}<\infty, \quad  \EP{\widehat
      X_T\widehat{Y}_T} = \widehat X_0 \widehat Y_0. 
  \end{displaymath}
  Then $\widehat X$ solves~\eqref{eq:3} for $x=\widehat X_0$ and
  $\widehat Y$ solves~\eqref{eq:6} for $y=\widehat Y_0$.
\end{Lemma}

\begin{proof}
  The result follows immediately from the identity
  \begin{displaymath}
    U(\widehat X_T) = V(\widehat Y_T) + \widehat X_T \widehat Y_T
  \end{displaymath}
  and the inequalities
  \begin{align*}
    U(X_T) & \leq V(\widehat Y_T) + X_T \widehat Y_T,
    \quad X\in \mathcal{X},  \\
    U(\widehat X_T) & \leq V(Y_T) + \widehat X_T Y_T, \quad Y\in
    \mathcal{Y},
  \end{align*}
  after we recall that $XY$ is a supermartingale for all $X\in
  \mathcal{X}$ and $Y\in \mathcal{Y}$.
\end{proof}

The goal of the paper is to find sufficient conditions for the lower
bound in~\eqref{eq:6} to be attained at some $\mathbb{Q}(y)\in
\mathcal{Q}$ called the \emph{optimal martingale measure} or,
equivalently, for the dual minimizer $Y(y)$ to be a \emph{uniformly
  integrable martingale}; in this case,
\begin{displaymath}
  Y_T(y) = y\frac{d\mathbb{Q}(y)}{d\mathbb{P}}.
\end{displaymath}
Our criteria are stated in Theorem~\ref{th:2} below, where a key role
is played by the probabilistic version of the classical Muckenhoupt
$\Ap{p}$ condition.

\section{$\Ap{p}$ condition for the dual minimizer}
\label{sec:ap-condition-dual}

Following~\cite[Section~2.3]{Kazam:94}, we recall the probabilistic
$\Ap{p}$ condition.

\begin{Definition}
  \label{def:1}
  Let $p>1$. An optional process $R\geq 0$ \emph{satisfies $\Ap{p}$}
  if $R_T>0$ and there is a constant $\C>0$ such that for every
  stopping time $\tau$
  \begin{displaymath}
    \cEP{\tau}{\left(\frac{R_\tau}{R_T}\right)^{\frac1{p-1}}} 
    \leq \C. 
  \end{displaymath}
\end{Definition}

Observe that if $R$ satisfies $\Ap{p}$, then $R$ satisfies $\Ap{p'}$
for every $p'\geq p$.

An important consequence of the $\Ap{p}$ condition is a uniform
integrability property. For continuous local martingales this fact is
well known and can be found e.g., in \cite[Section~2.3]{Kazam:94}.

\begin{Lemma}
  \label{lem:2}
  If an optional process $R\geq 0$ satisfies $\Ap{p}$ for some $p>1$
  and $\EP{R_T}<\infty$, then $R$ is of class $(\mathbf{D})$:
  \begin{displaymath}
    \descr{R_\tau}{\tau \text{ is a stopping time}} \text{ is
      uniformly integrable}. 
  \end{displaymath}
\end{Lemma}

\begin{proof}
  Let $\tau$ be a stopping time. As $p>1$, the function $x\mapsto
  x^{-\frac1{p-1}}$ is convex. Hence, by Jensen's inequality,
  \begin{displaymath}
    \cEP{\tau}{\left(\frac{R_\tau}{R_T}\right)^{\frac1{p-1}}}
    = R_\tau^{\frac1{p-1}}
    \cEP{\tau}{R_T^{-\frac1{p-1}}} \geq
    R_\tau^{\frac1{p-1}}
    \left(\cEP{\tau}{R_T}\right)^{-\frac1{p-1}}.
  \end{displaymath}
  Using the constant $\C>0$ from $\Ap{p}$, we obtain that
  \begin{displaymath}
    R_\tau \leq \C^{p-1} \cEP{\tau}{R_T},
  \end{displaymath}
  and the result follows.
\end{proof}

To motivate the use of the $\Ap{p}$ condition in the study of the dual
minimizers $Y(y)$, $y>0$, we first consider the case of power utility
with a positive power.

\begin{Proposition}
  \label{prop:1}
  Let~\eqref{eq:2} hold. Assume that
  \begin{displaymath}
    U(x) = \frac{x^{1-a}}{1-a}, \quad x>0,
  \end{displaymath}
  with the relative risk-aversion $a\in (0,1)$ and denote $p \set
  \frac1{1-a}>1$.  Then for $y>0$, the solution $Y(y)$ to the dual
  problem~\eqref{eq:4} exists if and only if
  \begin{equation}
    \label{eq:7}
    \EP{Y_T^{-\frac1{p-1}}} < \infty \mmtext{for some} Y\in
    \mathcal{Y}  
  \end{equation}
  and, in this case, for every $Y\in \mathcal{Y}$, $Y>0$ and every
  stopping time $\tau$,
  \begin{displaymath}
    \cEP{\tau}{\left(\frac{{Y}_\tau(y)}{{Y}_T(y)}\right)^{\frac1{p-1}}} 
    \leq \cEP{\tau}{\left(\frac{Y_\tau}{Y_T}\right)^{\frac1{p-1}}}.
  \end{displaymath}  
  In particular, $Y(y)$ satisfies $\Ap{p}$ if and only if there is
  $Y\in \mathcal{Y}$ satisfying $\Ap{p}$.
\end{Proposition}

\begin{proof}
  Observe that the convex conjugate to $U$ is given by
  \begin{displaymath}
    V(y) = \frac{a}{1-a} y^{-\frac{1-a}{a}} = (p-1) y^{-\frac1{p-1}}, \quad y>0. 
  \end{displaymath}
  Then~\eqref{eq:7} is equivalent to~\eqref{eq:5}, which, in turn, is
  equivalent to the existence of the optimal $Y(y)$, $y>0$. Denote
  $\widehat Y \set Y(1)$. Clearly, $Y(y) = y\widehat Y$.

  Let a stopping time $\tau$ and a process $Y\in \mathcal{Y}$, $Y>0$,
  be such that
  \begin{displaymath}
    \cEP{\tau}{\left(\frac{Y_\tau}{Y_T}\right)^{\frac1{p-1}}}<\infty. 
  \end{displaymath}
  We have to show that
  \begin{displaymath}
    \xi \set \cEP{\tau}{\left(\frac{\widehat
          Y_\tau}{\widehat Y_T}\right)^{\frac1{p-1}}} -
    \cEP{\tau}{\left(\frac{Y_\tau}{Y_T}\right)^{\frac1{p-1}}}
    \leq 0.  
  \end{displaymath}

  For a set $A\in \mathcal{F}_\tau$, the process
  \begin{displaymath}
    Z_t \set \widehat Y_t \ind{t\leq \tau} + \widehat Y_{\tau}
    \left(\frac{Y_t}{Y_\tau}\setind{A} + \frac{\widehat
        Y_t}{\widehat Y_\tau}(1 -
      \setind{A})\right) \ind{t>\tau}, \quad t\in [0,T],  
  \end{displaymath}
  belongs to $\mathcal{Y}$ and is such that $Z_0=1$ and $Z_\tau =
  \widehat Y_\tau$. We obtain that
  \begin{align*}
    \cEP{\tau}{\left(\frac{Z_\tau}{Z_T}\right)^{\frac1{p-1}}} &=
    \cEP{\tau}{\left(\frac{Y_\tau}{Y_T}\right)^{\frac1{p-1}}}
    \setind{A} + \cEP{\tau}{\left(\frac{\widehat Y_\tau}{\widehat
          Y_T}\right)^{\frac1{p-1}}} (1 -
    \setind{A}) \\
    &= \cEP{\tau}{\left(\frac{\widehat Y_\tau}{\widehat
          Y_T}\right)^{\frac1{p-1}}} - \xi \setind{A}.
  \end{align*}
  Dividing both sides by $Z_\tau^{\frac1{p-1}} = {\widehat
    Y}_\tau^{\frac1{p-1}}$ and choosing $A = \braces{\xi \geq 0}$, we
  deduce that
  \begin{displaymath}
    \EP{\left(\frac1{Z_T}\right)^{\frac1{p-1}}} =
    \EP{\left(\frac1{\widehat
          Y_T}\right)^{\frac1{p-1}}} -
    \EP{\left(\frac1{\widehat
          Y_\tau}\right)^{\frac1{p-1}} \max(\xi,0)}.
  \end{displaymath}
  However, the optimality of $\widehat Y=Y(1)$ implies that
  \begin{displaymath}
    \EP{\left(\frac1{\widehat
          Y_T}\right)^{\frac1{p-1}}} \leq 
    \EP{\left(\frac1{Z_T}\right)^{\frac1{p-1}}}.
  \end{displaymath}
  Hence $\xi\leq 0$.
\end{proof}

We now state the main result of the section.

\begin{Theorem}
  \label{th:1}
  Let~\eqref{eq:2} hold. Suppose that there are constants $0<a<1$,
  $b\geq a$ and $\C>0$ such that
  \begin{equation}
    \label{eq:8}
    \frac1{\C} \left(\frac{y}{x}\right)^a \leq \frac{U'(x)}{U'(y)} \leq
    \C \left(\frac{y}{x}\right)^b, \quad x\leq y, 
  \end{equation}
  and there is a supermartingale $Z\in \mathcal{Y}$ satisfying
  $\Ap{p}$ with
  \begin{displaymath}
    p = \frac1{1-a}. 
  \end{displaymath}
  Then for every $y>0$, the solution $Y(y)$ to~\eqref{eq:4} exists and
  satisfies $\Ap{p'}$ with
  \begin{displaymath}
    p' = 1 + \frac{b}{1-a}.
  \end{displaymath}
\end{Theorem}

\begin{Remark}
  \label{rem:1}
  Notice that if the relative risk-aversion of $U$ is well-defined and
  bounded away from $0$ and $\infty$, then in~\eqref{eq:8} we can take
  $C=1$ and choose $a$ and $b$ as lower and upper bounds:
  \begin{displaymath}
    0< a \leq  -\frac{xU''(x)}{U'(x)} \leq b<\infty, \quad x>0. 
  \end{displaymath}
  In particular, if
  \begin{displaymath}
    1 \leq -\frac{xU''(x)}{U'(x)} \leq b, \quad x>0,
  \end{displaymath}
  then choosing $a\in (0,1)$ sufficiently close to $1$ we fulfill the
  conditions of Theorem~\ref{th:1} if there exists a supermartingale
  $Z\in \mathcal{Y}$ satisfying $\Ap{p}$ for \emph{some} $p>1$.

  Observe also that for the positive power utility function $U$ with
  relative risk-aversion $a\in (0,1)$ we can select $b=a$ and then
  obtain same estimate as in Proposition~\ref{prop:1}:
  \begin{displaymath}
    p' = 1 + \frac{a}{1-a} = \frac1{1-a}=p.
  \end{displaymath} 
\end{Remark}

The proof of Theorem~\ref{th:1} relies on the following lemma.

\begin{Lemma}
  \label{lem:3}
  Assume~\eqref{eq:2} and suppose that there are constants $0<a<1$ and
  $\C_1>0$ such that
  \begin{equation}
    \label{eq:9}
    \frac1{\C_1} \left(\frac{y}{x}\right)^a \leq \frac{U'(x)}{U'(y)},
    \quad x\leq y, 
  \end{equation}
  and there is a supermartingale $Z\in \mathcal{Y}$ satisfying
  $\Ap{p}$ with
  \begin{displaymath}
    p = \frac1{1-a}. 
  \end{displaymath}
  Then for every $y>0$ the solution $Y(y)$ to~\eqref{eq:4} exists, and
  there is a constant $\C_2>0$ such that for every stopping time
  $\tau$ and every $y>0$,
  \begin{equation}
    \label{eq:10}
    \cEP{\tau}{I(Y_T(y))Y_T(y)}  \leq \C_2 I(Y_\tau(y))Y_\tau(y),  
  \end{equation}
  where $I=-V'$.
\end{Lemma}

\begin{Remark}
  \label{rem:2}
  Recall that for $x=-v'(y)$ the optimal wealth process $X(x)$ has the
  terminal value
  \begin{displaymath}
    X_T(x) = -V'(Y_T(y)) = I(Y_T(y))
  \end{displaymath}
  and the product $X(x)Y(y)$ is a uniformly integrable martingale. It
  follows that for every stopping time $\tau$
  \begin{displaymath}
    X_\tau(x) = \frac1{Y_\tau(y)} \cEP{\tau}{I(Y_T(y))Y_T(y)}
  \end{displaymath}
  and therefore, inequality~\eqref{eq:10} is equivalent to
  \begin{displaymath}
    X_\tau(x) \leq \C_2 I(Y_\tau(y)). 
  \end{displaymath}
\end{Remark}

\begin{proof}[Proof of Lemma~\ref{lem:3}.]
  To show the existence of $Y(y)$ we need to verify~\eqref{eq:5}.  As
  $I =-V'$ is the inverse function to $U$, condition~\eqref{eq:9} is
  equivalent to
  \begin{equation}
    \label{eq:11}
    \frac{I(x)}{I(y)} \leq \C_3 \left(\frac{y}{x}\right)^{1/a},
    \quad x\leq y,  
  \end{equation}
  where $\C_3 = \C_1^{1/a}$. From~\eqref{eq:11} we deduce that for
  $y\leq 1$
  \begin{align*}
    V(y) &= V(1) + \int_y^1 I(t) dt \leq V(1) + \C_3I(1)\int_y^1
    t^{-1/a} dt \\
    &= V(1) + \C_3I(1)\frac{a}{1-a}
    (y^{-\frac{1-a}{a}}-1) \\
    &= V(1) + \C_3I(1)(p-1) (y^{-\frac{1}{p-1}}-1).
  \end{align*}
  Hence, there is a constant $\C_4>0$ such that
  \begin{displaymath}
    V(y) \leq  \C_4(1+y^{-\frac1{p-1}}), \quad y>0. 
  \end{displaymath}
  As $Z$ satisfies $\Ap{p}$, we have
  \begin{displaymath}
    \EP{Z_T^{-\frac1{p-1}}} < \infty.
  \end{displaymath}
  It follows that
  \begin{displaymath}
    v(y) \leq \EP{V(yZ_T/Z_0)} < \infty, \quad y>0, 
  \end{displaymath}
  which completes the proof of the existence of $Y(y)$.

  Let $\tau$ be a stopping time and let $y>0$. We set $\widehat Y \set
  Y(y)$ and define the process
  \begin{displaymath}
    Y_t \set \widehat Y_t \ind{t\leq \tau} + \widehat Y_\tau
    \frac{Z_t}{Z_\tau} \ind{t>\tau}, \quad t\in [0,T]. 
  \end{displaymath}
  Clearly, $Y\in \mathcal{Y}$ and $Y_0 = \widehat Y_0 = y$.  We
  represent
  \begin{displaymath}
    I(\widehat Y_T) \widehat Y_T = \xi_1 + \xi_2 + \xi_3,
  \end{displaymath}
  by multiplying the left-side on the elements of the unity
  decomposition:
  \begin{displaymath}
    1 = \ind{\widehat Y_\tau \leq \widehat Y_T} + \ind{Y_T \leq \widehat
      Y_T < \widehat Y_\tau} + \ind{\widehat Y_T < Y_T, \widehat Y_T <
      \widehat Y_\tau}.
  \end{displaymath}

  For the first term, since $I = -V'$ is a decreasing function, we
  have that
  \begin{displaymath}
    \xi_1 = I(\widehat Y_T)\widehat Y_T \ind{\widehat Y_\tau \leq
      \widehat Y_T} \leq 
    I(\widehat Y_\tau) \widehat Y_T.
  \end{displaymath}
  Using the supermartingale property of $\widehat Y$, we obtain that
  \begin{displaymath}
    \cEP{\tau}{\xi_1} \leq  I(\widehat Y_\tau) \widehat Y_\tau. 
  \end{displaymath}

  For the second term, we deduce from~\eqref{eq:11} that
  \begin{align*}
    \xi_2 &= I(\widehat Y_T)\widehat Y_T \ind{Y_T \leq \widehat Y_T <
      \widehat Y_\tau} = I(\widehat Y_T) {\widehat Y_T}^{\frac1{a}}
    {\widehat Y_T}^{-\frac{1-a}{a}} \ind{Y_T \leq \widehat Y_T <
      \widehat Y_\tau} \\
    &\leq \C_3 I(\widehat Y_\tau) {\widehat Y_\tau}^{\frac1{a}}
    {Y_T}^{-\frac{1-a}{a}} = \C_3 I(\widehat Y_\tau) \widehat Y_\tau
    \left(\frac{Z_\tau}{Z_T}\right)^{\frac{1-a}{a}} = \C_3 I(\widehat
    Y_\tau) \widehat Y_\tau
    \left(\frac{Z_\tau}{Z_T}\right)^{\frac1{p-1}}
  \end{align*}
  and the $\Ap{p}$ condition for $Z$ yields the existence of a
  constant $\C_5>0$ such that
  \begin{displaymath}
    \cEP{\tau}{\xi_2} \leq \C_5  I(\widehat
    Y_\tau) \widehat Y_\tau. 
  \end{displaymath}

  For the third term, we deduce from~\eqref{eq:11} that
  \begin{align*}
    \xi_3 &= I(\widehat Y_T)\widehat Y_T \ind{\widehat Y_T < Y_T,
      \widehat Y_T < \widehat Y_\tau} \leq I(\widehat Y_T)\widehat Y_T
    \ind{\widehat
      Y_T < Y_T} \\
    &= I(\widehat Y_T)^{a} \widehat Y_T I(\widehat Y_T)^{1-a}
    \ind{\widehat Y_T < Y_T} \leq \C_1 I(Y_T)^{a} Y_T I(\widehat
    Y_T)^{1-a}
    \\
    & = \C_1(I(Y_T) Y_T)^{a} (I(\widehat Y_T) Y_T)^{1-a}
  \end{align*}
  and then from H\"older's inequality that
  \begin{displaymath}
    \cEP{\tau}{\xi_3} \leq \C_1(\cEP{\tau}{I(Y_T) Y_T})^{a}
    \left(\cEP{\tau}{I(\widehat Y_T) Y_T}\right)^{1-a}.
  \end{displaymath}

  We recall that the terminal wealth of the optimal investment
  strategy with $\widehat X_0 = -v'(y)$ is given by
  \begin{displaymath}
    I(\widehat Y_T) = \widehat X_T. 
  \end{displaymath}
  It follows that
  \begin{align*}
    \cEP{\tau}{I(\widehat Y_T) Y_T} &= \cEP{\tau}{\widehat X_T
      Y_T} \leq  \widehat X_\tau Y_\tau = \widehat X_\tau \widehat Y_\tau \\
    & = \cEP{\tau}{\widehat X_T \widehat Y_T} = \cEP{\tau}{I(\widehat
      Y_T) \widehat Y_T}.
  \end{align*}

  To estimate $\cEP{\tau}{I(Y_T) Y_T}$ we decompose
  \begin{displaymath}
    I(Y_T) Y_T = I(Y_T) Y_T \ind{\widehat Y_\tau \leq Y_T} + I(Y_T) Y_T
    \ind{\widehat Y_\tau > Y_T}. 
  \end{displaymath}
  Since $I$ is decreasing, we have that
  \begin{displaymath}
    I(Y_T) Y_T \ind{\widehat Y_\tau \leq Y_T} \leq I(\widehat Y_\tau) Y_T.
  \end{displaymath} 
  As $Y$ is a supermartingale and $Y_\tau = \widehat Y_\tau$, we
  obtain that
  \begin{displaymath}
    \cEP{\tau}{I(Y_T) Y_T \ind{\widehat Y_\tau \leq Y_T}} \leq
    I(\widehat Y_\tau) Y_\tau = I(\widehat Y_\tau) \widehat Y_\tau.  
  \end{displaymath}
  For the second term, using~\eqref{eq:11} we deduce that
  \begin{align*}
    I(Y_T) Y_T \ind{\widehat Y_\tau > Y_T} &= I(Y_T) Y_T^{\frac1a}
    Y_T^{-\frac{1-a}a} \ind{\widehat Y_\tau > Y_T} \leq \C_3
    I(\widehat Y_\tau)
    {\widehat Y_\tau}^{\frac1a} Y_T^{-\frac{1-a}a} \\
    &= \C_3 I(\widehat Y_\tau) \widehat Y_\tau \left(\frac{\widehat
        Y_\tau}{Y_T}\right)^{\frac{1-a}a} = \C_3 I(\widehat Y_\tau)
    \widehat Y_\tau \left(\frac{Z_\tau}{Z_T}\right)^{\frac{1}{p-1}}
  \end{align*}
  and the $\Ap{p}$ condition for $Z$ implies that
  \begin{displaymath}
    \cEP{\tau}{I(Y_T) Y_T \ind{\widehat Y_\tau > Y_T}} \leq \C_5
    I(\widehat Y_\tau) \widehat Y_\tau.  
  \end{displaymath}
  Thus we have
  \begin{displaymath}
    \cEP{\tau}{I(Y_T) Y_T} \leq \eta \set (1+ \C_5)
    I(\widehat Y_\tau) \widehat Y_\tau  
  \end{displaymath}
  and then
  \begin{displaymath}
    \cEP{\tau}{\xi_3} \leq C_1 \eta^{a}
    \left(\cEP{\tau}{I(\widehat Y_T) \widehat Y_T}\right)^{1-a}. 
  \end{displaymath}

  Adding together the estimates for $\cEP{\tau}{\xi_i}$ we obtain that
  \begin{displaymath}
    \cEP{\tau}{I(\widehat Y_T) \widehat Y_T} \leq \eta + C_1 \eta^{a}
    \left(\cEP{\tau}{I(\widehat Y_T) \widehat Y_T}\right)^{1-a}.
  \end{displaymath}
  It follows that
  \begin{displaymath}
    \cEP{\tau}{I(\widehat Y_T) \widehat Y_T} \leq x^* \eta =
    x^*(1+\C_5) I(\widehat Y_\tau) \widehat Y_\tau,
  \end{displaymath}
  where $x^*$ is the root of
  \begin{displaymath}
    x = 1 + C_1 x^{1-a}, \quad x>0.  
  \end{displaymath}
  We thus have proved inequality~\eqref{eq:10} with
  $\C_2=(1+\C_5)x^*$.
\end{proof}

\begin{proof}[Proof of Theorem~\ref{th:1}.]
  Fix $y>0$. In view of Lemma~\ref{lem:3}, we only have to verify that
  $\widehat Y \set Y(y)$ satisfies $\Ap{p'}$.

  Denote $\widehat X \set X(-v'(y))$ and recall that by
  Lemma~\ref{lem:3} and Remark~\ref{rem:2}, there is $\C_2>0$ such
  that, for every stopping time $\tau$,
  \begin{displaymath}
    \widehat X_\tau \leq \C_2 I(\widehat Y_\tau).   
  \end{displaymath}
  Observe also that as $I =-V'$ is the inverse function to $U'$, the
  second inequality in~\eqref{eq:8} is equivalent to
  \begin{displaymath}
    \frac{y}{x} \leq \C \left(\frac{I(x)}{I(y)}\right)^{b}, \quad x\leq y.
  \end{displaymath}

  We fix a stopping time $\tau$. Since $I(\widehat Y_T) = \widehat
  X_T$, we deduce from the inequalities above that
  \begin{displaymath}
    \left(\frac{\widehat Y_\tau}{\widehat Y_T}\right)^{1/b} \leq
    \max\left(1,  \C^{1/b} 
      \frac{I(\widehat Y_T)}{I(\widehat Y_\tau)}\right)
    \leq \max\left(1,  \C_3
      \frac{\widehat X_T}{\widehat X_\tau}\right),
  \end{displaymath}
  where $\C_3 = \C^{1/b}\C_2$.  It follows that
  \begin{align*}
    \left(\frac{\widehat Y_\tau}{\widehat Y_T}\right)^{\frac1{p'-1}}
    &= \left(\frac{\widehat Y_\tau}{\widehat
        Y_T}\right)^{\frac{1-a}{b}} \leq \max\left(1, \C_3^{1-a}
      \left(\frac{\widehat X_T}{\widehat X_\tau}\right)^{1-a}\right)
    \\
    &\leq 1 + \C_3^{1-a} \left(\frac{\widehat X_TZ_T}{\widehat X_\tau
        Z_\tau}\right)^{1-a}\left(\frac{Z_\tau}{Z_T}\right)^{1-a}.
  \end{align*}
  Denoting by $\C_1>0$ the constant in the $\Ap{p}$ condition for $Z$,
  we deduce from H\"older's inequality and the supermartingale
  property of $\widehat X Z$ that
  \begin{align*}
    \cEP{\tau}{\left(\frac{\widehat Y_\tau}{\widehat
          Y_T}\right)^{\frac1{p'-1}}} &\leq 1 + \C_3^{1-a}
    \left(\cEP{\tau}{\frac{\widehat X_TZ_T}{\widehat X_\tau
          Z_\tau}}\right)^{1-a}
    \left(\cEP{\tau}{\left(\frac{Z_\tau}{Z_T}\right)^{\frac1{p-1}}}\right)^a
    \\ &\leq 1 + \C_3^{1-a} \C_1^a.
  \end{align*}
  Hence, $\widehat Y$ satisfies $\Ap{p'}$.
\end{proof}

\section{Local martingale property for maximal elements of
  $\mathcal{Y}$}
\label{sec:local-mart-prop}

Even if the dual minimizer $Y(y)$ is uniformly integrable, it may not
be a martingale, due to the lack of the local martingale property; see
the single-period example for logarithmic utility in
\cite[Example~5.1$'$]{KramSch:99}. Proposition~\ref{prop:2} below
yields sufficient conditions for every \emph{maximal} element of
$\mathcal{Y}$ (in particular, for $Y(y)$) to be a local martingale.

A semimartingale $R$ is called \emph{$\sigma$-bounded} if there is a
predictable process $h>0$ such that the stochastic integral $\int hdR$
is bounded. Following \cite{KramSirb:06}, we make the following
assumption.

\begin{Assumption}
  \label{as:1}
  For all $X$ and $X'$ in $\mathcal{X}$ such that $X>0$, the process
  $X'/X$ is $\sigma$-bounded.
\end{Assumption}

Assumption~\ref{as:1} holds easily if stock price $S$ is
continuous. Theorem 3 in Appendix of \cite{KramSirb:06} provides a
sufficient condition in the presence of jumps. It states that
\emph{every} semimartingale $R$ is $\sigma$-bounded if there is a
finite-dimensional local martingale $M$ such that every bounded purely
discontinuous martingale $N$ is a stochastic integral with respect to
$M$.

\begin{Proposition}
  \label{prop:2}
  Suppose that Assumption~\ref{as:1} holds. Let $Y\in \mathcal{Y}$ be
  such that $YX'$ is a local martingale for some $X'\in \mathcal{X}$,
  $X'>0$. Then $YX$ is a local martingale for every $X\in
  \mathcal{X}$. In particular, $Y$ is a local martingale.
\end{Proposition}

\begin{proof}
  We assume first that $X'=Y=1$.  Let $X\in \mathcal{X}$. As $X$ is
  $\sigma$-bounded, there is a predictable $h>0$ such that
  \begin{displaymath}
    \abs{\int hdX} \leq 1. 
  \end{displaymath}
  Since the bounded non-negative processes $1\pm\int hdX$ belong to
  $\mathcal{X}$, they are supermartingales, which is only possible if
  $\int hdX$ is a martingale. It follows that $X$ is a non-negative
  stochastic integral with respect to a martingale:
  \begin{displaymath}
    X = X_0 + \int \frac1h d(\int hdX) \geq 0. 
  \end{displaymath}
  Therefore, $X$ is a local martingale, see~\cite{AnselStr:94}.  Under
  the condition $X'=Y=1$, the proof is obtained.

  We now consider the general case. Without loss of generality, we can
  assume that $X'_0 = Y_0=1$. By localization, we can also assume that
  the local martingale $YX'$ is uniformly integrable and then define a
  probability measure $\mathbb{Q}$ with the density
  \begin{displaymath}
    \frac{d\mathbb{Q}}{d\mathbb{P}} = X'_T Y_T. 
  \end{displaymath}
  Let $X\in\mathcal{X}$. We have that $XY$ is a local martingale under
  $\mathbb{P}$ if and only if $X/X'$ is a local martingale under
  $\mathbb{Q}$.

  By Assumption~\ref{as:1}, the process $X/X'$ is
  $\sigma$-bounded. Elementary computations show that $X/X'$ is a
  wealth process in the financial market with stock price
  \begin{displaymath}
    S' = \left(\frac1{X'}, \frac{S}{X'}\right);
  \end{displaymath}
  see \cite{DelbSch:95}. The result now follows by applying the
  previous argument to the $S'$-market whose reference probability
  measure is given by $\mathbb{Q}$.
\end{proof}

\section{Existence of the optimal martingale measure}
\label{sec:exist-optimal-mart}

Recall that $X(x)$ denotes the optimal wealth process for the primal
problem~\eqref{eq:3}, while $Y(y)$ stands for the minimizer to the
dual problem~\eqref{eq:4}. As usual, the \emph{density process} of a
probability measure $\mathbb{R} \ll \mathbb{P}$ is a uniformly
integrable martingale (under $\mathbb{P}$) with the terminal value
$\frac{d\mathbb{R}}{d\mathbb{P}}$.

The following is the main result of the paper.

\begin{Theorem}
  \label{th:2}
  Let Assumption~\ref{as:1} hold. Suppose that there are constants
  $0<a<1$, $b\geq a$ and $\C>0$ such that
  \begin{equation}
    \label{eq:12}
    \frac1{\C} \left(\frac{y}{x}\right)^a \leq \frac{U'(x)}{U'(y)} \leq
    \C \left(\frac{y}{x}\right)^b, \quad x\leq y, 
  \end{equation}
  and there is a martingale measure $\mathbb{Q}\in\mathcal{Q}$ whose
  density process $Z$ satisfies $\Ap{p}$ with
  \begin{equation}
    \label{eq:13}
    p = \frac1{1-a}. 
  \end{equation}
  Then for every $y>0$ the optimal martingale measure $\mathbb{Q}(y)$
  exists and its density process $Y(y)/y$ satisfies $\Ap{p'}$ with
  \begin{displaymath}
    p' = 1 + \frac{b}{1-a}.
  \end{displaymath}
\end{Theorem}

\begin{proof}
  From Theorem~\ref{th:1} we obtain that the dual minimizer $Y(y)$
  exists and satisfies $\Ap{p'}$ and then from Lemma~\ref{lem:2} that
  it is of class $(\mathbf{D})$.  The local martingale property of
  $Y(y)$ follows from Proposition~\ref{prop:2}, if we account for
  Assumption~\ref{as:1} and the martingale property of
  $X(-v'(y))Y(y)$. Thus, $Y(y)$ is a uniformly integrable martingale
  and hence, $Y(y)/y$ is the density process of the optimal martingale
  measure $\mathbb{Q}(y)$.
\end{proof}

We refer the reader to Remark~\ref{rem:1} for a discussion of the
conditions of Theorem~\ref{th:2}.

\begin{Example}
  \label{ex:1}
  In a typical situation, the role of the ``testing'' martingale
  measure $\mathbb{Q}$ is played by the \emph{minimal} martingale
  measure, that is, by the optimal martingale measure for logarithmic
  utility. For a model of stock prices driven by a Brownian motion,
  its density process $Z$ has the form:
  \begin{displaymath}
    Z_t = \mathcal{E}\left(-\lambda \cdot B\right)_t := \exp\left(-\int_0^t \lambda
      dB - \frac12 \int_0^t \abs{\lambda_s}^2 ds \right), \quad t\in
    [0,T],
  \end{displaymath}
  where $B$ is an $N$-dimensional Brownian motion and $\lambda$ is a
  predictable $N$-dimensional process of the \emph{market price of
    risk}. We readily deduce that $Z$ satisfies $\Ap{p}$ for all $p>1$
  if both $\lambda$ and the maturity $T$ are bounded. This fact
  implies the assertions of Theorem~\ref{th:2}, provided that
  inequalities~\eqref{eq:12} hold for some $a\in (0,1)$, $b\geq a$ and
  $C>0$ or, in particular, if the relative risk-aversion of $U$ is
  bounded away from $0$ and $\infty$.
\end{Example}

The following result shows that the key bound~\eqref{eq:13} is the
best possible.

\begin{Theorem}
  \label{th:3}
  Let constants $a$ and $p$ be such that
  \begin{displaymath}
    0<a<1 \mmtext{and} p>\frac1{1-a}.
  \end{displaymath}
  Then there exists a financial market with a continuous stock price
  $S$ such that
  \begin{enumerate}
  \item There is a $\mathbb{Q}\in\mathcal{Q}$ whose density process
    $Z$ satisfies $\Ap{p}$.
  \item In the optimal investment problem with the power utility
    function
    \begin{displaymath}
      U(x) = \frac{x^{1-a}}{1-a}, \quad x>0, 
    \end{displaymath}
    the dual minimizers $Y(y) = y\widehat Y$, $y>0$, are well-defined,
    but are not uniformly integrable martingales. In particular, the
    optimal martingale measure $\widehat{\mathbb{Q}} = \mathbb{Q}(y)$
    does not exist.
  \end{enumerate}
\end{Theorem}

The proof of Theorem~\ref{th:3} follows from Proposition~\ref{prop:3}
below, which contains an exact counterexample.

We conclude the section with a couple of corollaries of
Theorem~\ref{th:2} which exploit connections between the $\Ap{p}$
condition and $\bmo$ martingales. Hereafter, we shall refer to
\cite{Kazam:94} and therefore, restrict ourselves to the continuous
case.

\begin{Assumption}
  \label{as:2}
  All local martingales on the filtered probability space $(\Omega,
  \mathcal{F}, (\mathcal{F}_t)_{t\in [0,T]}, \mathbb{P})$ are
  continuous.
\end{Assumption}

From Assumption~\ref{as:2} we deduce that the density process of every
$\mathbb{Q}\in \mathcal{Q}$ is a continuous uniformly integrable
martingale and that the dual minimizer $Y(y)$ is a {continuous} local
martingale.

We recall that a continuous local martingale $M$ with $M_0=0$ belongs
to $\bmo$ if there is a constant $\C>0$ such that
\begin{equation}
  \label{eq:14}
  \cEP{\tau}{\qv{M}_T - \qv{M}_\tau} \leq \C\; \mtext{for every stopping
    time} \tau,
\end{equation}
where $\qv{M}$ is the quadratic variation process for $M$. It is known
that $\bmo$ is a Banach space with the norm
\begin{displaymath}
  \bmonorm{M} \set \inf\descr{\sqrt{\C}>0}{\eqref{eq:14}\;\text{holds for }
    C>0}.  
\end{displaymath}
We also recall that for a continuous local martingale $M$ with
$M_0=0$,
\begin{enumerate}[label=(\roman{*}), ref=(\roman{*})]
\item \label{item:1} The stochastic exponential $\mathcal{E}(M) \set
  e^{M - \qv{M}/2}$ satisfies $\Ap{p}$ for \emph{some} $p>1$ if and
  only if $M\in\bmo$; see Theorem~2.4 in \cite{Kazam:94}.
\item \label{item:2} The stochastic exponentials $\mathcal{E}(M)$ and
  $\mathcal{E}(-M)$ satisfy $\Ap{p}$ for \emph{all} $p>1$ if and only
  the martingale
  \begin{equation}
    \label{eq:15}
    q(M)_t \set \cEP{t}{\qv{M}_T} - \EP{\qv{M}_T},
    \quad t\in[0,T],
  \end{equation}
  is well-defined and belongs to the closure in $\bmonorm{\cdot}$ of
  the space of bounded martingales; see Theorem~3.12 in
  \cite{Kazam:94}.
\end{enumerate}

\begin{Corollary}
  \label{cor:1}
  Let Assumption~\ref{as:2} hold. Suppose that there are constants
  $b\geq 1$ and $\C>0$ such that
  \begin{equation}
    \label{eq:16}
    \frac1{\C} \left(\frac{y}{x}\right) \leq \frac{U'(x)}{U'(y)} \leq
    \C \left(\frac{y}{x}\right)^b, \quad x\leq y, 
  \end{equation}
  and there is a martingale measure $\mathbb{Q}\in\mathcal{Q}$ with
  density process $Z = \mathcal{E}(M)$ with $M\in \bmo$.  Then for
  every $y>0$ the optimal martingale measure $\mathbb{Q}(y)$ exists
  and its density process is given by $Y(y)/y=\mathcal{E}(M(y))$ with
  $M(y)\in \bmo$.
\end{Corollary}

\begin{proof}
  From \ref{item:1} we deduce that $Z$ satisfies $\Ap{p}$ for some
  $p>1$. Clearly, \eqref{eq:16} implies \eqref{eq:12} for every $a\in
  (0,1)$ and in particularly for $a$
  satisfying~\eqref{eq:13}. Theorem~\ref{th:2} then implies that
  $Y(y)/y$ satisfies $\Ap{p'}$ for some $p'>1$ and another
  application of~\ref{item:1} yields the result.
\end{proof}

We notice that by~\ref{item:1} and Theorem~\ref{th:3} the power $1$ in
the first inequality of~\eqref{eq:16} cannot be replaced with any
$a\in (0,1)$, in order to guarantee that the optimal martingale
measure $\mathbb{Q}(y)$ exists.

\begin{Corollary}
  \label{cor:2}
  Let Assumption~\ref{as:2} hold and let inequality \eqref{eq:12} be
  satisfied for some constants $0<a<1$, $b\geq a$ and $\C>0$. Suppose
  also that there is a martingale measure $\mathbb{Q}\in\mathcal{Q}$
  whose density process $Z = \mathcal{E}(M)$ is such that the
  martingale $q(M)$ in~\eqref{eq:15} is well-defined and belongs to
  the closure in $\bmonorm{\cdot}$ of the space of bounded
  martingales.  Then for every $y>0$ the optimal martingale measure
  $\mathbb{Q}(y)$ exists and its density process is given by
  $Y(y)/y=\mathcal{E}(M(y))$ with $M(y)\in \bmo$.
\end{Corollary}

\begin{proof}
  The result follows directly from \ref{item:2} and
  Theorem~\ref{th:2}.
\end{proof}

\section{Counterexample}
\label{sec:counterexample}

In this section we construct an example of financial market satisfying
the conditions of Theorem~\ref{th:3}.  For a semimartingale $R$, we
denote by $\mathcal{E}(R)$ its stochastic exponential, that is, the
solution of the linear equation:
\begin{displaymath}
  d\mathcal{E}(R) = \mathcal{E}(R)_{-} dR, \quad \mathcal{E}(R)_0 =
  1. 
\end{displaymath}

We start with an auxiliary filtered probability space $(\Omega,
\mathcal{F}, (\mathcal{F}_t)_{t\geq 0}, \mathbb{Q})$, which supports a
Brownian motion $B=(B_t)$ and a counting process $N=(N_t)$ with the
stochastic intensity $\lambda = (\lambda_t)$ given in~\eqref{eq:19}
below; $B_0=N_0=0$.  We define the process
\begin{displaymath}
  S_t \set \mathcal{E}(B)_t = e^{B_t - t/2}, \quad t\geq 0,
\end{displaymath}
and the stopping times
\begin{align*}
  T_1 & \set \inf\descr{t\geq 0}{S_t=2},  \\
  T_2 & \set \inf\descr{t\geq 0}{N_t=1},  \\
  T & \set T_1 \wedge T_2 = \min(T_1,T_2).
\end{align*}

We fix constants $a$ and $p$ such that
\begin{equation}
  \label{eq:17}
  0<a<1 \mmtext{and} p>\frac1{1-a}
\end{equation}
and choose a constant $b$ such that
\begin{equation}
  \label{eq:18}
  a < b < \frac1{q} \mmtext{and}  \gamma \leq \frac12 \delta(1-\delta),
\end{equation}
where
\begin{align*}
  q &\set \frac{p}{p-1} < \frac1a, \\
  \delta &\set b-a >0, \\
  \gamma &\set \frac{b}2(1-qb) > 0.
\end{align*}

With this notation, we define the stochastic intensity $\lambda =
(\lambda_t)$ as
\begin{equation}
  \label{eq:19}
  \lambda_t \set  \frac{\gamma}{1-(S_t/2)^\delta}
  \ind{t<T_1} + \gamma \ind{t\geq T_1}, \quad t\geq 0.  
\end{equation}
Recall that $N - \int \lambda dt$ is a local martingale under
$\mathbb{Q}$.

Finally, we introduce a probability measure $\mathbb{P} \ll
\mathbb{Q}$ with the density
\begin{displaymath}
  \frac{d\mathbb{P}}{d\mathbb{Q}} = \frac{1}{\EQ{S_T^b}} S_T^b.  
\end{displaymath}
Notice that
\begin{equation}
  \label{eq:20}
  \braces{ \frac{d\mathbb{P}}{d\mathbb{Q}} = 0} = \braces{S_T=0} 
  = \braces{\mathcal{E}(B)_T = 0} = \braces{T=\infty}
\end{equation}
and therefore, the stopping time $T$ is finite under $\mathbb{P}$:
\begin{displaymath}
  \mathbb{P}(T<\infty) = 1.
\end{displaymath}  

\begin{Proposition}
  \label{prop:3}
  Assume~\eqref{eq:17} and~\eqref{eq:18} and consider the financial
  market with the price process $S$ and the maturity $T$ defined on
  the filtered probability space $(\Omega, \mathcal{F}_T,
  (\mathcal{F}_t)_{t\in [0,T]}, \mathbb{P})$. Then
  \begin{enumerate}
  \item The probability measure $\mathbb{Q}$ belongs to $\mathcal{Q}$
    and the density process $Z$ of $\mathbb{Q}$ with respect to
    $\mathbb{P}$ satisfies $\Ap{p}$.
  \item In the optimal investment problem with the power utility
    function
    \begin{equation}
      \label{eq:21}
      U(x) = \frac{x^{1-a}}{1-a}, \quad x>0, 
    \end{equation}
    the dual minimizers $Y(y) = y\widehat Y$, $y>0$, are well-defined
    but are not uniformly integrable martingales. In particular, the
    optimal martingale measure $\widehat{\mathbb{Q}} = \mathbb{Q}(y)$
    does not exist.
  \end{enumerate}
\end{Proposition}

The proof is divided into a series of lemmas.

\begin{Lemma}
  \label{lem:4}
  The stopping time $T$ is finite under $\mathbb{Q}$ and the
  probability measures $\mathbb{P}$ and $\mathbb{Q}$ are equivalent.
\end{Lemma}

\begin{proof}
  In view of~\eqref{eq:20}, we only have to show that
  \begin{displaymath}
    \mathbb{Q}(T<\infty)=1.  
  \end{displaymath}
  Indeed, by~\eqref{eq:19}, the intensity $\lambda$ is bounded below
  by $\gamma>0$ and hence,
  \begin{displaymath}
    \mathbb{Q}(T>t) \leq  \mathbb{Q}(T_2>t) \leq e^{-\gamma t} \to 0, 
    \quad t\to \infty.  
  \end{displaymath}
\end{proof}

From the construction of the model and Lemma~\ref{lem:4} we deduce
that $\mathbb{Q}\in \mathcal{Q}$. To show that the density process $Z$
of $\mathbb{Q}$ with respect to $\mathbb{P}$ satisfies $\Ap{p}$ we
need the following estimate.

\begin{Lemma}
  \label{lem:5}
  Let $0<\epsilon<1$ be a constant and $\tau$ be a stopping time. Then
  \begin{displaymath}
    \cEQ{\tau}{S^\epsilon_T} \leq S^\epsilon_\tau \leq \left(1 +
      \frac{\epsilon(1-\epsilon)}{2\gamma}\right) 
    \cEQ{\tau}{S^\epsilon_T}.
  \end{displaymath}
\end{Lemma}

\begin{proof}
  We denote
  \begin{displaymath}
    \theta = \frac12 \epsilon(1-\epsilon)
  \end{displaymath}
  and deduce that
  \begin{displaymath}
    S^\epsilon_t = \mathcal{E}(B)^\epsilon_t = \mathcal{E}(\epsilon B)_t
    e^{-\theta t}, \quad t\in [0,T]. 
  \end{displaymath}
  In particular, $S^\epsilon$ is a $\mathbb{Q}$-supermartingale, and
  the first inequality in the statement of the lemma follows.

  To verify the second inequality, we define local martingales $L$ and
  $M$ under $\mathbb{Q}$ as
  \begin{align*}
    L_t &= \int_0^t \frac{\theta}{\lambda_r} (dN_r -
    \lambda_r dr),  \\
    M_t &= \mathcal{E}(\epsilon B)_t \mathcal{E}(L)_t,
  \end{align*}
  and observe that
  \begin{align*}
    M_t &=  S^\epsilon_t, \quad t\leq T, \; t<T_2, \\
    M_T &= \left(1 + \frac{\theta}{\lambda_T}\right) S^\epsilon_T,
    \quad T=T_2.
  \end{align*}
  Since $\lambda \geq \gamma$, we obtain that
  \begin{displaymath}
    S^\epsilon_t \leq M_t \leq \left(1 + \frac{\theta}{\gamma}\right)
    S^\epsilon_t, \quad t\in [0,T].   
  \end{displaymath}
  As $S\leq 2$, we deduce that $M$ is a bounded
  $\mathbb{Q}$-martingale and the result readily follows.
\end{proof}

\begin{Lemma}
  \label{lem:6}
  The density process $Z$ of $\mathbb{Q}$ with respect to $\mathbb{P}$
  satisfies $\Ap{p}$.
\end{Lemma}

\begin{proof}
  Fix a stopping time $\tau$. As $\mathbb{Q}\sim \mathbb{P}$, we have
  \begin{align*}
    \cEP{\tau}{\left(\frac{Z_\tau}{Z_T}\right)^{\frac1{p-1}}} &=
    \cEQ{\tau}{\left(\frac{Z_\tau}{Z_T}\right)^{1+\frac1{p-1}}} =
    \cEQ{\tau}{\left(\frac{Z_\tau}{Z_T}\right)^{q}} \\
    &= \cEQ{\tau}{\left(\frac{\widetilde Z_T}{\widetilde
          Z_\tau}\right)^{q}},
  \end{align*}
  where $\widetilde Z = 1/Z$ is the density process of $\mathbb{P}$
  with respect to $\mathbb{Q}$.

  Recall that
  \begin{displaymath}
    \widetilde Z_T = C S_T^b, 
  \end{displaymath}
  for some constant $\C>0$. Since $0<b<bq<1$, Lemma~\ref{lem:5} yields
  that
  \begin{align*}
    \widetilde Z_\tau = \cEQ{\tau}{\widetilde Z_T} &= \C
    \cEQ{\tau}{S^b_T}
    \geq \C \left(1 + \frac{b(1-b)}{2\gamma}\right)^{-1} S^b_\tau, \\
    \cEQ{\tau}{\widetilde{Z}^q_T} &= \C^q \cEQ{\tau}{S^{qb}_T} \leq
    \C^q S^{qb}_\tau,
  \end{align*}
  which implies the result.
\end{proof}

We now turn our attention to the second item of
Proposition~\ref{prop:3}. Of course, our financial market has been
specially constructed in such a way that the solutions $X(x)$ and
$Y(y)$ to the primal and dual problems are quite explicit.

\begin{Lemma}
  \label{lem:7}
  In the optimal investment problem with the utility function $U$
  from~\eqref{eq:21}, it is optimal to buy and hold stocks:
  \begin{displaymath}
    X(x) = xS, \quad x>0. 
  \end{displaymath}
  The dual minimizers have the form $Y(y) = y\widehat Y$, $y>0$, with
  \begin{equation}
    \label{eq:22}
    \widehat Y = \mathcal{E}(L)Z, 
  \end{equation} 
  where $Z$ is the density process of $\mathbb{Q}$ with respect to
  $\mathbb{P}$ and
  \begin{equation}
    \label{eq:23}
    L_t =  \int_0^{t} \frac{\gamma}{\lambda_r} (\lambda_r dr - dN_r),
    \quad t\in [0,T]. 
  \end{equation}
\end{Lemma}

\begin{proof}
  We verify the conditions of Lemma~\ref{lem:1}.  For the stochastic
  exponential $\mathcal{E}(L)$ we obtain that
  \begin{displaymath}
    \mathcal{E}(L)_t = e^{\gamma t}, \quad t<T, 
  \end{displaymath}
  and, as $S_{T_1} = 2$, that
  \begin{align*}
    \mathcal{E}(L)_T &= e^{\gamma T} \left(\ind{ T = T_1} + \left(1-
        \frac{\gamma}{\lambda_T}\right)\ind{T =
        T_2}\right) \\
    &= e^{\gamma T} \left(\ind{ T = T_1} +
      \left(\frac{S_T}{2}\right)^\delta\ind{T =
        T_2}\right) \\
    &= e^{\gamma T} \left(\frac{S_T}{2}\right)^\delta.
  \end{align*}
  Hence for $\widehat Y$ defined by~\eqref{eq:22} we have
  \begin{displaymath}
    \widehat Y_T = \mathcal{E}(L)_TZ_T = \C S_T^{-a} = \C U'(S_T), 
  \end{displaymath}
  for some constant $\C>0$.

  Let $X\in \mathcal{X}$.  Under $\mathbb{Q}$, the product
  $X\mathcal{E}(L)$ is a local martingale, because $X$ is a stochastic
  integral with respect to the Brownian motion $B$ and
  $\mathcal{E}(L)$ is a purely discontinuous local martingale. It
  follows that $X\widehat Y = X\mathcal{E}(L)Z$ is a non-negative
  local martingale (hence, a supermartingale) under
  $\mathbb{P}$. Thus,
  \begin{displaymath}
    \widehat Y \in \mathcal{Y}. 
  \end{displaymath}

  Observe that the convex conjugate to $U$ is given by
  \begin{displaymath}
    V(y) = \frac{a}{1-a} y^{-\frac{1-a}{a}}, \quad y>0. 
  \end{displaymath}
  It follows that
  \begin{displaymath}
    V(y\widehat Y_T) = V(y) \widehat Y_T \widehat Y_T^{-1/a} = V(y)\C^{-1/a}
    \widehat Y_T S_T 
  \end{displaymath}
  and therefore,
  \begin{displaymath}
    \EP{V(y\widehat Y_T)} \leq V(y)\C^{-1/a} < \infty, \quad
    y>0. 
  \end{displaymath}

  To conclude the proof we only have to show that the local martingale
  $S\widehat Y=S\mathcal{E}(L)Z$ under $\mathbb{P}$ is of class
  $(\mathbf{D})$ or, equivalently, that the local martingale
  $S\mathcal{E}(L)$ under $\mathbb{Q}$ is of class $(\mathbf{D})$.
  Actually, we have a stronger property:
  \begin{displaymath}
    \descr{S_\tau\mathcal{E}(L)_\tau}{\tau \text{ is a stopping time}}
    \mtext{is bounded in} \mathbf{L}^q(\mathbb{Q}).
  \end{displaymath} 
  Indeed,
  \begin{displaymath}
    S_t\mathcal{E}(L)_t \leq S_t e^{\gamma t} \leq 2^{1-b} S_t^b
    e^{\gamma t}, \quad t\in [0,T], 
  \end{displaymath}
  and then for a stopping time $\tau$,
  \begin{align*}
    \EQ{\left(S_\tau\mathcal{E}(L)_\tau\right)^q} &\leq
    2^{q(1-b)}\EQ{\left(S^b_\tau e^{\gamma\tau}\right)^q} =
    2^{q(1-b)}\EQ{\mathcal{E}(B)^{qb}_\tau e^{q\gamma\tau}}
    \\
    &= 2^{q(1-b)}\EQ{\mathcal{E}(qbB)_\tau} \leq 2^{q(1-b)}.
  \end{align*}
\end{proof}

The following lemma completes the proof of the proposition.

\begin{Lemma}
  \label{lem:8}
  For the dual minimizer $\widehat Y$ constructed in Lemma~\ref{lem:7}
  we have
  \begin{displaymath}
    \EP{\widehat Y_T} < 1.
  \end{displaymath}
  Thus, $\widehat Y$ is not a uniformly integrable martingale.
\end{Lemma}

\begin{proof}
  Recall from the proof of Lemma~\ref{lem:7} that for the local
  martingale $L$ defined in~\eqref{eq:23},
  \begin{displaymath}
    \mathcal{E}(L)_T = e^{\gamma T} \left(\frac{S_T}{2}\right)^\delta.
  \end{displaymath}
  Using~\eqref{eq:18}, we deduce that
  \begin{displaymath}
    \mathcal{E}(L)_T = 
    \frac1{2^\delta} e^{\gamma T} \left(\mathcal{E}(B)_T\right)^\delta = 
    \frac1{2^\delta} e^{\gamma T} \mathcal{E}(\delta B)_T e^{-\frac12
      \delta(1-\delta)T} \leq \frac1{2^\delta} \mathcal{E}(\delta B)_T.
  \end{displaymath}
  It follows that
  \begin{displaymath}
    \EP{\widehat Y_T} = \EP{\mathcal{E}(L)_T Z_T} 
    = \EQ{\mathcal{E}(L)_T} \leq \frac1{2^\delta}
    \EQ{\mathcal{E}(\delta B)_T} \leq  \frac1{2^\delta}. 
  \end{displaymath} 
\end{proof}

\bibliographystyle{plainnat}

\bibliography{../bib/finance}

\end{document}